\documentclass[a4paper,aps,twocolumn,nobibnotes,prl]{revtex4-1}

\usepackage[T1]{fontenc}
\usepackage[utf8]{inputenc}
\usepackage{lmodern}
\usepackage{graphicx}
\usepackage{dcolumn}
\usepackage{bm}
\usepackage{amsmath}
\usepackage{amstext}
\usepackage{amssymb}
\usepackage{amsthm}
\usepackage{dsfont}
\usepackage{color}
\usepackage[caption=false]{subfig}

\newcommand{\ket}[1]{\ensuremath{| #1 \rangle}}

\newcommand{\eq}[1]{Eq.~(\ref{eq:#1})}

\newcommand{\be}{\begin{equation}}
\newcommand{\ee}{\end{equation}}

\newcommand{\bea}{\begin{eqnarray}}
\newcommand{\eea}{\end{eqnarray}}

\newcommand{\mean}[1]{\ensuremath{\langle{#1}\rangle}}

\renewcommand{\rho}{\varrho}
\renewcommand{\tilde}{\widetilde}

\theoremstyle{definition}
\newtheorem{theorem}{Theorem}

\newcommand{\Span}{\ensuremath \text{span}}
\newcommand{\rank}{\ensuremath \text{rank}}
\renewcommand{\vec}[1]{\ensuremath \bm{#1}}

\newcommand{\cone}{\text{cone}}
\newcommand{\tr}{\text{Tr}}

\newcommand{\kommentar}[1]{}

\usepackage{xr}
\begin{document}

\title{Generalizing optimal Bell inequalities}

\author{Fabian Bernards}
\author{Otfried Gühne}%
\affiliation{
Naturwissenschaftlich-Technische Fakultät, 
Universität Siegen, 
Walter-Flex-Straße 3, 
57068 Siegen, Germany}

\date{\today}

\begin{abstract}
Bell inequalities are central tools for studying nonlocal correlations 
and their applications in quantum information processing. Identifying 
inequalities for many particles or measurements is, however, difficult 
due to the computational complexity of characterizing the set of local 
correlations. We develop a method to characterize Bell inequalities under 
constraints, which may be given by symmetry or other linear conditions. This 
allows to search systematically for generalizations of given Bell 
inequalities to more parties. As an example, we find all possible 
generalizations of the two-particle inequality by Froissart 
[Il Nuovo Cimento {\bf B64}, 241 (1981)], also known as I3322 inequality, 
to three particles. For the simplest of these inequalities, we study
their quantum mechanical properties and demonstrate that they are relevant, 
in the sense that they detect nonlocality of quantum states, for which 
all two-setting inequalities fail to do so. 
\end{abstract}

\maketitle

{\it Introduction.---}
Bell nonlocality describes the fact that quantum mechanical 
correlations are stronger than the ones predicted by local
hidden variable (LHV) models \cite{brunnerbellreview, scarani2019}. 
This manifests itself in the violation of Bell inequalities, which 
have been observed experimentally \cite{shalm2015,hensen2015,giustina2015, rosenfeld2017}. 
Besides ruling out hidden variable models, however, Bell inequalities are 
also essential for studying information theory with distributed parties 
\cite{buhrmann2010}. In fact, various Bell inequalities found interesting 
applications and led to surprising insights to quantum information processing. 
Examples are the connection between Bell inequalities and communication complexity 
\cite{brukner2004, tavakoli2019}, multi-player games which demonstrate the difference
between quantum mechanics and non-signaling theories \cite{almeida2010},
Bell inequalities that are useful for multiparty conference key agreement 
\cite{holz2019} and the emerging field of self-testing, where Bell 
inequalities play also a central role, see \cite{supic2019} for a review.
Therefore, although LHV models are considered to be experimentally refuted, 
it is desirable to identify Bell inequalities with interesting properties. 

Characterizing all the Bell inequalities is, however, not straight forward. 
For a given number of parties, if one fixes the number of measurement
settings and the outcomes per measurement, the set of all probabilities 
coming from LHV models forms a high-dimensional polytope and the Bell 
inequalities correspond to the facets of this polytope 
\cite{peres1999, pitowsky1991}. The extremal points of this polytope are 
easy to characterize, but it is computationally very demanding to find all 
the facets from the extremal points \cite{pitowsky1991}. The rising complexity 
can be directly illustrated with an example. If one considers two parties 
with two measurements ($A_1, A_2$ and $B_1, B_2$) and two outcomes ($\pm 1$), 
it is well known that there is, up to relabelings
and permutations, only one optimal Bell inequality \cite{fine1982}, known as the 
Clauser-Horne-Shimony-Holt (CHSH) inequality \cite{clauser1969, clauser1970erratum}. 
It reads
\begin{align}
\mean{A_1 B_1} +
\mean{A_1 B_2} +
\mean{A_2 B_1} -
\mean{A_2 B_2} 
\le 2.
\label{eq-chsh}
\end{align}
Some generalizations of it, using more particles but still two measurements 
per site, have been found \cite{werner2001,zukowski2002, sliwa2003}.

If one considers two particles with three measurements, the analysis becomes 
already considerably harder. It was shown that there is only one additional 
Bell inequality. This was first identified by Froissart \cite{froissart1981}, 
and later independently by \'Sliwa \cite{sliwa2003} and Collins and Gisin 
\cite{collins2004} \footnote{Collins and Gisin also presented a complete 
list of facet defining Bell inequalities for the scenario, in which Alice 
has four and Bob has three settings.}. It reads
\begin{align}
\mean{A_1} & - \mean{A_2} + \mean{B_1} - \mean{B_2} - 
\mean{(A_1 - A_2)(B_1 - B_2)}
\nonumber
\\ 
& + \mean{(A_1 + A_2)B_3} + \mean{A_3(B_1 + B_2)}
\le 4.
\label{eq-i3322}
\end{align}
This inequality, henceforth called I3322 inequality, has several interesting 
properties: It detects the nonlocality of some two-qubit states which are not 
detected by the CHSH inequality \cite{collins2004}, and, for higher dimensions, 
the maximal violation of it is not attained at maximally entangled states 
\cite{vidick2011}. In general, the violation of I3322 is conjectured to 
increase with the dimension of the underlying quantum system \cite{pal2010}
and, since the maximal violation can definitely not occur in small-dimensional
systems \cite{moroder2013, navascues2015}, I3322 can be used for the device-independent 
characterization of the dimension. Clearly, it is highly desirable to find 
generalizations of I3322 to three or more particles, but the exponentially 
increasing complexity of a brute force approach has prevented this so far.

In this paper, we present a systematic approach to find all Bell inequalities 
obeying some linear constraints. This constraint may be given by a desired 
symmetry or it may be formulated such that the Bell inequality should reduce 
to some fixed expression in special cases. With our method, we find all 3050 
generalizations of the I3322 inequality to three particles, and then we 
characterize the physical properties of the simplest ones of the new 
inequalities. Our method generalizes previous methods to characterize Bell 
inequalities constrained by symmetry \cite{bancal2010} in a 
constructive manner, as it can be used to systematically find all 
generalizations of a known inequality.

{\it Bell inequalities and convex geometry.---}
A Bell scenario is characterized by the number of parties $N$, the 
number of measurement settings per party $I$ and the number of outcomes 
$O$ per measurement. For a given scenario, the behavior of a physical 
system  can be encoded in a vector that contains all the probabilities 
for the outcomes of all possible joint measurements on the subsystems. 
Concretely, the correlations of a system are encoded in the vector that 
contains all probabilities $p(o_1,...,o_N|i_1, ... , i_N)$, where 
$o_k \in \{1,...,O\}$ denotes the outcomes for the settings 
$i_k \in \{1,...,I\}$. In the simple case of 
dichotomic measurements (with outcomes $\pm 1$), the correlations can 
alternatively be characterized as the collection of the expectation 
values of joint measurements. This is also how the CHSH- and I3322 
inequality were formulated above.

We call systems that can be described by an LHV model classical. Such
systems can be characterized in the following way. First, one considers the 
finite set of probability vectors with local deterministic assignments, 
where each probability is $0$ or $1$ and the results of joint measurements 
on different systems are uncorrelated. Then, one considers all probabilistic 
mixtures of these vertices, i.e. the convex hull. The resulting set of all classical 
behaviors forms a convex polytope, the so-called local polytope. As any polytope, 
the local polytope is uniquely represented by its vertices as extremal points,
but one can alternatively describe it by linear inequalities.

In fact, a Bell inequality is nothing but a linear inequality in the space of the 
probabilities, which holds on the local polytope. So, it defines two half-spaces 
in the space of correlations, such that one half-space contains the local polytope and the other
half-space is detected as nonlocal. Clearly, the most efficient characterization 
of the local polytope is the minimal set of half-spaces whose intersection equals 
the polytope. These correspond to the optimal Bell inequalities and are characterized 
by the facets of the polytope \footnote{Note that for experimental tests other
notions of optimality may be advantageous, see \cite{vandam2005} for a discussion.}. It is easy to check whether a given linear inequality 
corresponds to a facet of the local polytope, as one just has to find enough vertices 
that saturate it. The conversion from the vertex description to the half-space
description of the local polytope is, however, difficult in practice and is 
known to be hard for the general case \cite{pitowsky1991}. 

{\it The main method.---}
We consider the situation where a convex polytope $P$ in vertex representation 
is given and one aims to find all facets of $P$ that satisfy some linear conditions.
The naive way of solving this problem is to first find all facets of $P$ and filter 
out all those inequalities that do not meet the conditions in a second step. In 
practice, however, already the first step is often to complex to be solved.
Our method to solve this problem works for the case that the conditions are 
affine equality constraints on the coefficients of the facet defining inequality.

\begin{figure}[t]
  \subfloat[\label{sfig:1a}]{\includegraphics[width=0.17\textwidth]{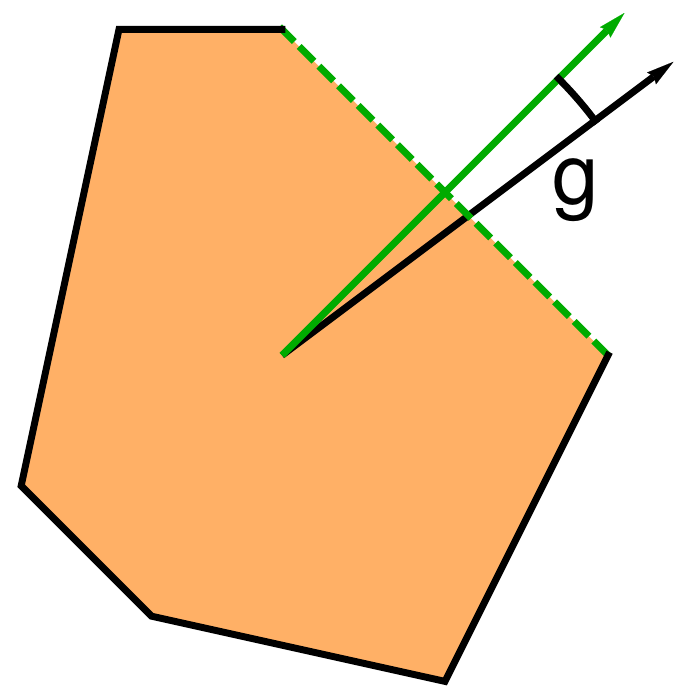}}
  \hspace{0.3cm}
  \subfloat[\label{sfig:1b}]{\includegraphics[width=0.25\textwidth]{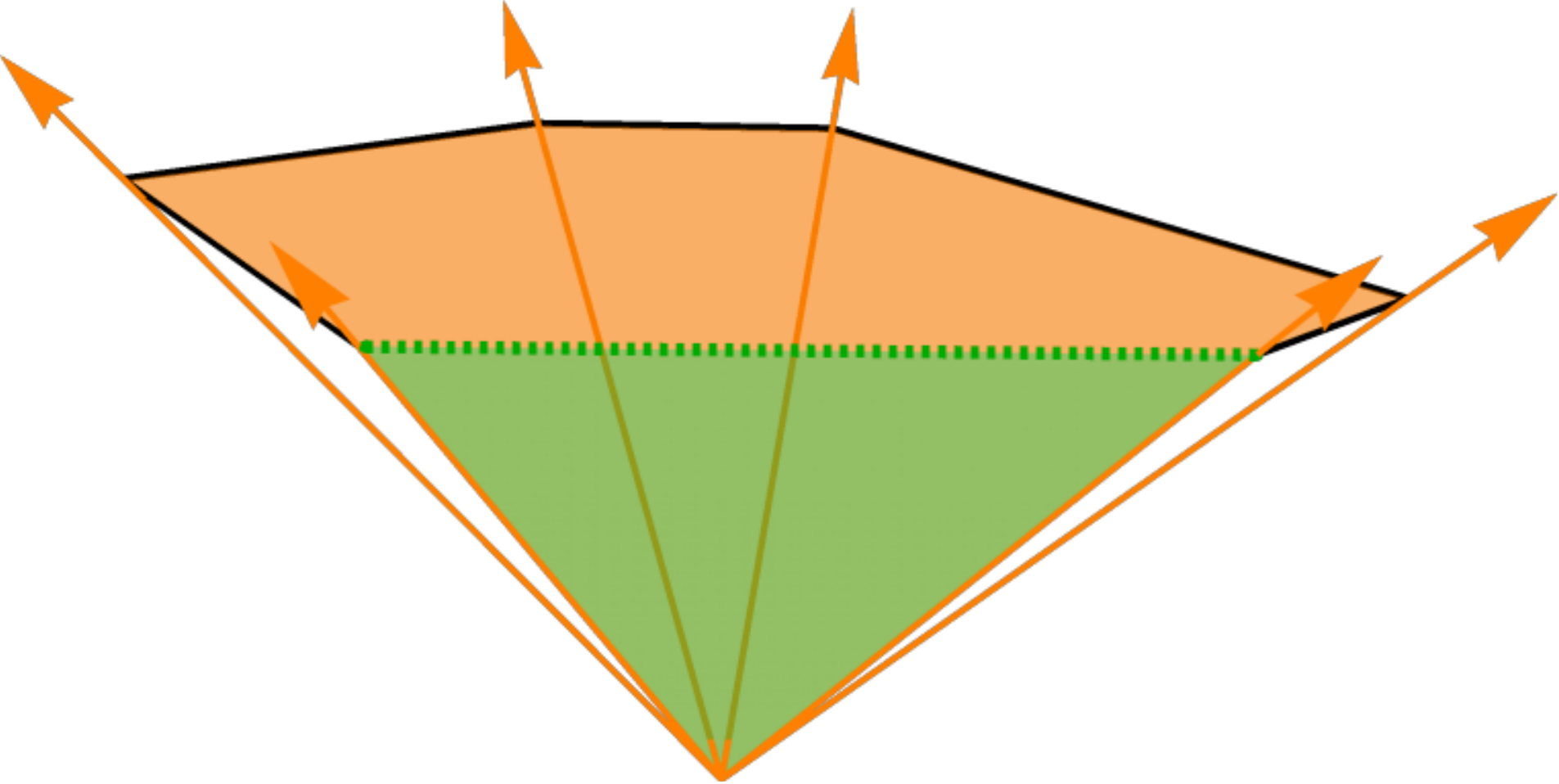}}
  
   \subfloat[\label{sfig:1c}]{\includegraphics[width=0.22\textwidth]{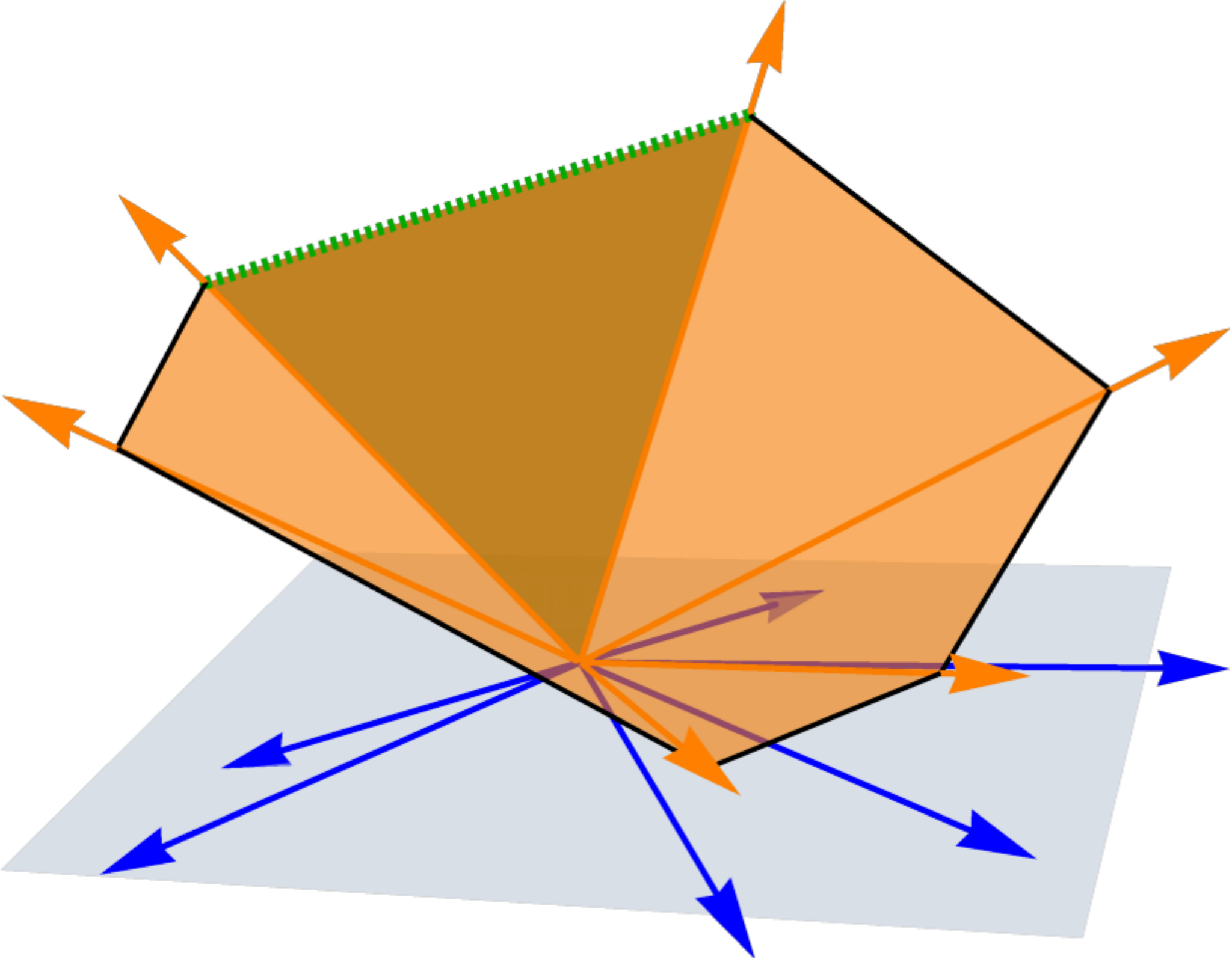}}
  \hspace{0.3cm}
  \subfloat[\label{sfig:1d}]{\includegraphics[width=0.21\textwidth]{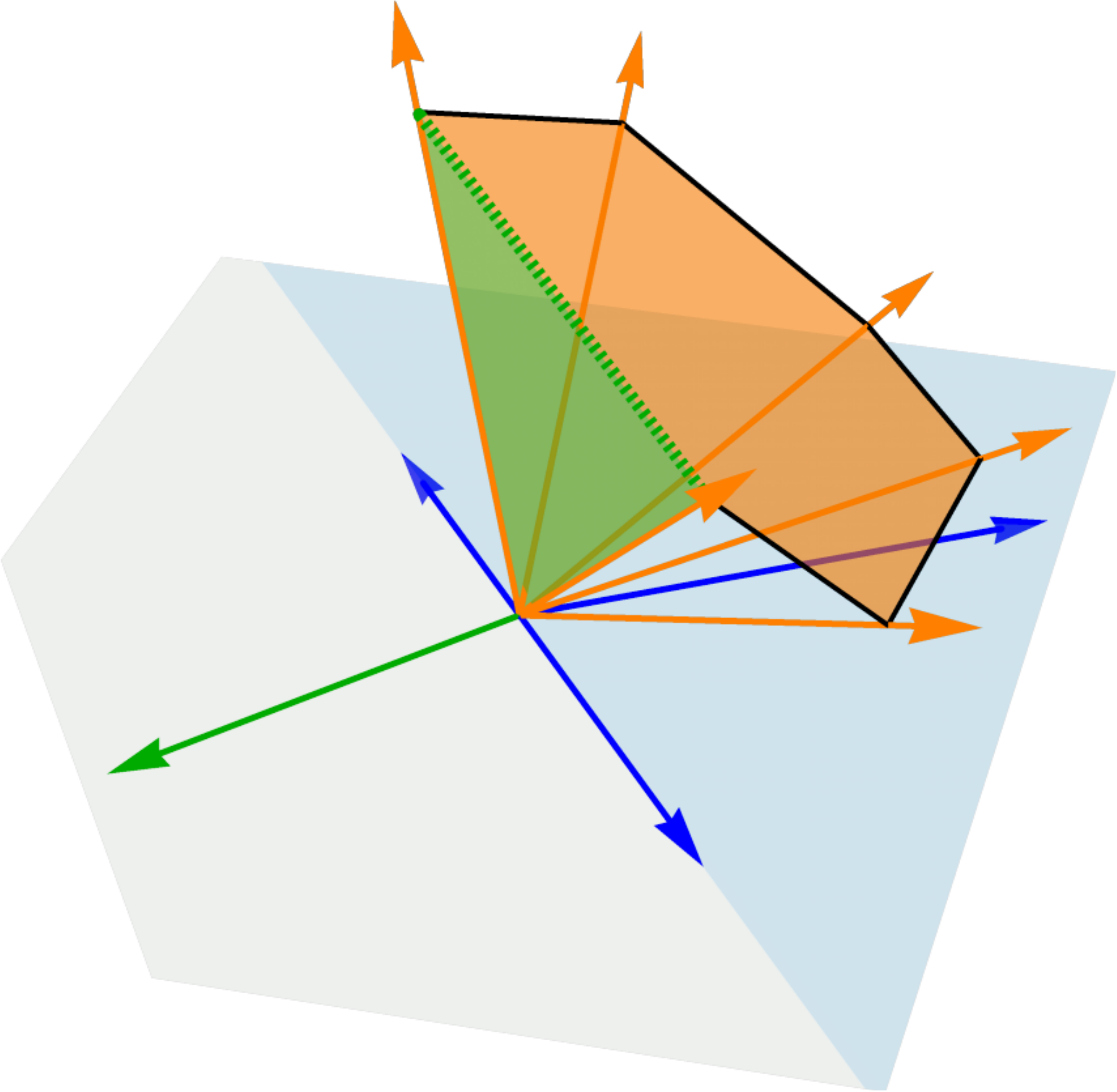}}
  
 \caption{Visualization of the method to find Bell inequalities
that obey some constraints. 
(a) We aim to find facets of the two-dimensional polytope $P$ with 
a normal vector $\vec{b}_P$ that has a fixed scalar product with 
some vector $\vec{g}$. The facet and the normal vector that fulfill 
this constraint drawn in green. 
(b) We embed the polytope $P$ in a plane in three-dimensional space. 
The coordinates of the vertices of $P$ define rays (orange arrows) that 
define the cone $C$. The polytope $P$ is then a section of the cone 
$C$ and each of its facets relates to a facet of $C$ (green) in a unique 
way. 
(c) The initial constraint on the facet of $P$ can be translated. 
A facet of $P$ fulfills the constraint if and only if the corresponding 
facet of $C$ has a normal vector $\vec{b}_C$ that obeys a linear
constraint $G \vec{b}_C=0$, which means that it has to lie in a
plane (light-grey). Then, we project the rays of $C$ (orange) into that 
plane (blue arrows) to define a cone $\tilde C$ as generated by the 
projected rays. By construction, facets of $C$ which obey the constraint
are also facets of $\tilde C$.
(d) Finally we find the facets of $\tilde C$ and check which ones 
correspond to facets of $C$. From the facets of $C$ that meet the 
conditions we can then compute the corresponding facets of $P$. 
In our example, $\tilde C$ is a half plane (light-blue) and has only one facet,
with the normal vector in green. It is also the normal vector of 
a facet of $C$ (green). Note that in the example $\tilde C$ 
is already generated by three rays and the other three rays are 
redundant.
}
\label{fig1}
\end{figure}

In the following, we describe the main ideas of the different steps, 
details are given in Appendix A \footnote{See the supplemental material, 
which includes references \cite{ziegler, cdd, wolf2009, abramsky2011,
horodecki1996,peres1996, mosek, picos}.}. We start with the original $D$-dimensional 
polytope $P$ and the constraints, which can be written in the form that the 
normal vector $\vec{b}_P$ on the facet has a fixed scalar product with some 
vector $\vec{g}$, see Fig.~\ref{fig1}(a). Then, we construct a cone $C$ in 
$D+1$-dimensional space that maintains a one-to-one correspondence to
the polytope, since the polytope can be seen as a cut of the cone, see 
Fig.~\ref{fig1}(b). Notably, there is a one-to-one correspondence between 
the facets of the cone and those of the polytope, and a normal vector 
$\vec{b}_P$ of a polytope facet translates to a normal vector $\vec{b}_C$ 
of a cone facet. Also, this construction allows to write the constraints in 
a linear form as $G \vec{b}_C=0$, see Fig.~\ref{fig1}(c).

The key observation is that in this situation we can define a new cone
$\tilde C$, such that if $\vec{b}_C$ is a facet normal vector of $C$ that obeys
the constraints, then $\vec{b}_C$ is also a facet normal vector of $\tilde C$. 
This is done by projecting the rays of $C$ down to the subspace of vectors
obeying $G \vec{v}=0$, see Fig.~\ref{fig1}(c). The advantage is that $\tilde C$ 
lies in a significantly lower-dimensional space. Additionally, $\tilde C$ has 
typically much less rays than $C$. That makes it easier to find all the 
facets of $\tilde C$, compared with $C$. Having obtained all the facets of
$\tilde C$, it remains to check whether they are facets of $C$. If this
is the case, they automatically obey the constraints, see 
Fig.~\ref{fig1}(d).

{\it Tripartite generalizations of I3322.---}
Let us apply our method to the scenario where three parties perform three 
dichotomic measurements. This scenario is already too complex to find all facets 
of the local polytope: 
The three-partite scenario with two dichotomic measurements per party has 64 
vertices in 26 dimensions and  53856 facets already \cite{sliwa2003}. Concerning the scenario we 
are interested in, we  only know that the local polytope has 512 vertices 
in a 63-dimensional space. Bancal and coworkers have investigated this scenario while 
restricting themselves to symmetric, full-body correlation inequalities 
and they found 20 facet defining Bell inequalities of this type \cite{bancal2010}. 
However, due the fact that marginal correlations prove to be vital
for the I3322 inequality, the restriction to Bell inequalities without them seems to be ad-hoc.
 
The notion of generalizing a bipartite Bell inequality to more parties is 
best introduced by example. Consider the Mermin inequality for three 
parties,
\begin{align}
\mean{A_1 B_1 C_2} + \mean{A_1 B_2 C_1} + \mean{A_2 B_1 C_1} - \mean{A_2 B_2 C_2}  
\leq 2. 
\label{eq-mermin}
\end{align}
If one assigns the fixed values $C_1=C_2=1$ (or some other values) on the 
third party, then this inequality reduces to the CHSH inequality in Eq.~(\ref{eq-chsh}) 
(or a variant thereof). In this sense, one may view the Mermin inequality as a generalization
of the CHSH inequality.

The formal definition of a generalization of the I3322 inequality is as follows. 
First, since all measurements are dichotomic, we write all Bell inequalities in 
terms of expectation values of observables. Observables that refer to measurements 
of party $A$ ($B$, $C$) are denoted as $A_i$ ($B_j$, $C_k$) and we define 
$A_0, B_0, C_0$ to be trivial measurements on the respective parties that 
always yield a measurement result $+1$. This conveniently allows to treat 
marginal terms such as $\mean{A_1}=\mean{A_1 B_0 C_0}$ and constant terms 
such as $1 = \mean{A_0 B_0 C_0}$ on the same footing. In this notation, 
any Bell inequality can be written as 
$S=\sum_{i,j,k} b_{ijk} \langle A_i B_j C_k \rangle \ge 0.$

We call a three-partite Bell inequality a generalization of the I3322 
inequality if the following three conditions hold: 
(a) The inequality is symmetric under exchange of the parties, each 
of which can perform three different non-trivial dichotomic measurements, 
(b) The inequality corresponds to a facet of the local polytope, and 
(c) There is an assignment $C_k \to \xi_k \in\{\pm 1\}$ for the 
observables on the third party, such that the remaining inequality 
$\sum_{i,j,k} b_{ijk} \xi_k \langle A_i B_j\rangle \ge 0$
is the I3322 inequality as in Eq.~(\ref{eq-i3322}). The analogous 
condition with the same $\xi_k$ holds for any other reduction to 
two parties. Note that this approach is different from the notion 
of lifting Bell inequalities to more parties \cite{pironio2005}, 
where, e.g., a two-party Bell inequality is applied to the conditional 
probabilities of a three-party system. 

The properties (a) and (b) are natural requirements, because both are 
characteristic features of the I3322 inequality. Besides its mathematical
motivation, condition (c) has physical consequences. It implies that the 
new inequality detects nonlocality in a three-partite state 
$\varrho_{ABC}$ every time the I3322 inequality detects it in the reduced
state $\varrho_{AB}.$ Indeed, if condition (c) holds, one can just
choose trivial measurements for the $C_k$ that yield outcomes $
\xi_k \in\{\pm 1\}$ independently of the state. Then the generalization of I3322 also detects
the nonlocality. 

For our method, we also note that condition (c) can be 
reformulated as follows: An inequality like 
$\sum_{i,j,k} b_{ijk} \xi_k \langle A_i B_j\rangle \ge 0$
is the I3322 inequality, if and only if equality holds for the 
set of vertices for which also equality holds in I3322. These bipartite 
vertices may be lifted to three-partite vertices by adding the variable
$\xi_k$, then any generalization of the I3322
inequality has to be saturated by these lifted vertices. The critical reader 
may ask at this point why we are only considering the special form of I3322 
as in Eq.~(\ref{eq-i3322}), and not equivalent forms arising from a relabeling 
of the observables or a sign flip. First, since we defined generalizations of 
I3322 to be symmetric, we only have to take symmetric versions of I3322 into 
account. Further, as one can easily check, all symmetric versions of I3322 can 
be transformed into each other just by outcome relabelings on both parties. 
Consequently, it suffices to only consider 
one symmetric version of I3322.

We can now find all generalizations of I3322 in four steps: First, we 
find all vertices of the local polytope of two parties which saturate 
I3322. Second, we choose deterministic outcomes $\xi_k$ on $C$ and determine 
the corresponding vertices in the local polytope of three parties. 
In the third step, we compose the matrix $G$ for the condition $G \vec{b}_C=0$.
This matrix contains each symmetry condition and each vertex from step 2 as
a row. Finally, we employ our method as described above to find all facet defining 
inequalities that meet the criteria. We then repeat steps two to four until all 
possible choices for deterministic outcomes $\xi_k$ are exhausted. In our case 
of three dichotomic measurements, there are eight possibilities. In this way, we 
find all symmetric, facet-defining generalizations of I3322, $3050$ inequalities 
in total. Details of the implementation are given in Appendix B and the complete 
list is given in the supplemental material.

 
{\it Properties of the generalizations of I3322.---}
Let us now examine  the three 
simplest ones among the generalized I3322 inequalities in some detail. For 
convenience, we introduce a short-hand 
notation for symmetric Bell inequalities. We define symmetric correlations 
as $(ijk)= \sum_{\pi\in \Pi} A_{\pi(i)} B_{\pi(j)} C_{\pi(k)}$, where $\Pi$ 
denotes the set of all permutations of the indices $ijk$ that give different 
terms. Note that in this notation $(112)= A_1 B_1 C_2 + A_1 B_2 C_1 + A_2 B_1 C_1$,  
so permutations leading to the same term are not counted multiple
times. Further, as noted before, settings labeled with index zero refer to 
trivial measurements that always yield the result $1$. 
Using this notation, the I3322 inequality in Eq.~(\ref{eq-i3322}) can be 
written as
\begin{align}
(01) - (02) - (11) - (22) + (12) + (13) + (23) \le 4.
\end{align}
The three generalizations of I3322 that involve the least number of symmetric
correlations are given by:
\begin{align}
F_1 & =   8  + (110) - (210) + (211) + (220) 
\nonumber \\  & + (222) -2 (331) -2 (332) \ge 0,
\label{eq-bi1} \\
F_2 &=  9  + (110) +2 (220) -2 (221) - (300) 
\nonumber \\ & - (310) + (311) -2 (322) \ge 0,
\label{eq-bi2} \\
F_3& = 9  - (210) + (211) + (220) +3 (222) 
\nonumber \\ &
- (300) - (310) + (311) -2 (322) \ge 0.
\label{eq-bi3}
\end{align}

Let us start our analysis with the possible violations in quantum mechanics. 
It has recently been shown that all inequalities that exclusively utilize full
correlations are violated in quantum mechanics \cite{escola2020}, however, our 
inequalities also contain marginal terms. 

A way to obtain bounds on possible quantum values of Bell inequalities
is given by the hierarchy of Navascués, Pironio and Ac\'{\i}n (NPA) 
\cite{navascues2007, navascues2008, wittek-package}. 
For the inequality $F_1$, this shows
that in quantum mechanics $F_1 \geq -8$ holds. The minimal value $F_1 = -8$ can 
indeed be reached, namely by a three-qubit GHZ state 
$\ket{GHZ_3}=(\ket{000}+\ket{111})/\sqrt{2}$ and measurements
$A_1 = A_2 = \sigma_x$, 
$A_3 = \sigma_y$, 
$B_1 = B_2 = -\sigma_y$, 
$B_3 =\sigma_x$ 
and $C_i = B_i$.
In fact, $F_1$ reduces to the Mermin inequality in Eq.~(\ref{eq-mermin}) if
the first two measurement settings on each party are chosen equal.

Concerning $F_2$, a numerical optimization suggests that one optimal 
choice of settings is given by 
$A_1 = -A_3 = \sigma_z$, 
$A_2 = -\sigma_x$, 
$B_1 = -B_3 = -\sigma_z$, 
$B_2 =-\sigma_x$, 
and $C_i = B_i$.
This leads to a quantum mechanical violation of $F_2=4(1-\sqrt 7)\approx-6.58301$
for the three-qubit state
\begin{align}
\ket{\psi_2}= a\ket{W_3} + b\ket{111},
\end{align}
with $\ket{W_3}=(\ket{001}+\ket{010}+\ket{100})/\sqrt{3}$ being the three-qubit
W state and 
$a=\sqrt{19+2\sqrt{7}}/\sqrt{74} \approx 0.57294$
and
$b=\sqrt{55-2\sqrt{7}}/\sqrt{74} \approx 0.81960.$ 
The violation attained by this state coincides up to numerical precision
with the lower bound on $F_2$ for quantum states from the NPA hierarchy.
It is interesting that $F_2$ is maximally violated by a state that does not
belong to the frequently studied three-qubit states (such as the states
considered in \cite{bruss1998, hein2004, guehne2014}). In this way, the Bell 
inequality $F_2$ may open an avenue for new methods of self-testing 
quantum states \cite{supic2019}.

For the inequality $F_3$ the third level of the NPA hierarchy bounds 
the quantum mechanical values by $F_3 \geq -4.63097.$ Within numerical 
precision, this can be attained using the three-qubit state 
\begin{equation}
|\psi_3 \rangle = \cos(\varphi) |W_3\rangle + \sin(\varphi) |GHZ_3\rangle 
\end{equation}
with $\varphi=4.0^{\circ}$. The required measurement settings (for $i=1,2,3$)
are
$A_i = \cos(\alpha_i) \sigma_z + \sin(\alpha_i) \sigma_x$
and
$B_i= C_i = \cos(\alpha_i) \sigma_z - \sin(\alpha_i) \sigma_x$,
where
$\alpha_1 = 141.6^{\circ}, \alpha_2 = 22.6^{\circ} , \alpha_3 = 101.6^{\circ} .$
Again, we find a non-standard three-qubit state leading to the maximal 
violation of the Bell inequality. The state $\ket{\psi_3}$ is close to 
a three-qubit W state but the difference is significant, as for the
W state one can only reach a violation of $F_3 = -4.59569$.

Now we clarify whether the new three-setting inequalities are indeed 
relevant, that is, whether they detect the nonlocality of some quantum 
states, where all two-setting inequalities fail to do so. This question 
can be answered positively for all {\it three} inequalities; moreover, 
all the $F_i$ detect also entanglement that is not 
detected by the I3322 inequality in the reduced states. To show this, 
we provide a three-qubit state 
$\varrho_{ABC}$ with separable two-body marginals that has a symmetric 
extension for each $F_i$, such that the respective Bell inequality $F_i$ is
violated. A symmetric extension of 
$\varrho_{ABC}$ is a five-qubit state $H_{ABB'CC'}$ that is symmetric 
under exchange of the parties  $B$, $B'$ (and $C$, $C'$) such that 
$\varrho_{ABC} =\tr_{B'C'}(H_{ABB'CC'})$. A state that has a such a 
symmetric extension cannot violate any Bell inequality with with an
arbitrary number of settings for Alice and two settings for Bob and 
Charlie \cite{terhal2003}, see also Appendix C. Note that here the number of
outcomes for each setting is unrestricted. Thus this is a stronger
statement than proving that the known inequalities for two settings
and two outcomes  \cite{werner2001,zukowski2002, sliwa2003} are not violated.

We find the desired state $\varrho_{ABC}$ using a seesaw algorithm 
that alternates between optimizing the measurement settings (for the
violation of the $F_i$) and the state (under some constraints).
If measurement settings for all parties are fixed, finding a 
state with a symmetric extension that maximally violates a given 
$F_i$ is a semidefinite program. On the other hand, given a state 
and measurement settings for two parties, finding the optimal 
settings for the third party is also a semidefinite program. 
Details and examples are discussed in Appendix C.

{\it Conclusion.---}
We presented a method to find all Bell inequalities for
a given Bell scenario under some constraints. Our method
does not require to characterize the local polytope 
completely, instead, all candidates for the desired Bell 
inequalities can be found in a low-dimensional projection
of the cone corresponding to the polytope. Using our method, 
we characterized all generalizations of the I3322 inequality
to three particles. It turned out that already the simplest ones of 
these generalizations have interesting properties, making
an experimental implementation of them desirable.

Our method can be used for many other purposes, where convex
polytopes play a role. First, one may consider other Bell scenarios, 
e.g., cases where not all parties have the same number of measurements. 
In addition, one can generalize also contextuality inequalities \cite{budroni2020} 
or recently discussed inequalities for testing certain views of quantum mechanics 
\cite{bong2019}. Moreover, it could be interesting to extend
our approach beyond the scenario of linear optimization: While 
the local polytope in Bell scenarios can be characterized by the optimization
method of linear programming, other forms of quantum correlations, such as
quantum steering, are naturally described in terms of convex
optimization and semidefinite programs \cite{cavalcanti2017, uola2020}. Thus, a 
generalization of our methods to this type of problems could
give more insight into various problems in information 
processing.

{\it Acknowledgements.---}
We thank Jędrzej Kaniewski, Miguel Navascués, Chau Nguyen and 
Denis Rosset for fruitful discussions. 
This work has been supported by the Deutsche Forschungsgemeinschaft (DFG, German 
Research Foundation - 447948357) and the ERC (Consolidator Grant 
683107/TempoQ). FB acknowledges support from the House of Young Talents 
of the University of Siegen.

\section{Appendix A: Details of the method}
\label{Appmethod}

\subsection{Some terminology for convex polyhedra}

We follow the standard terminology as introduced in the textbook 
by Ziegler \cite{ziegler} and refer to convex polyhedra simply as 
polyhedra, since we are only concerned with the convex variety.
The two arguably most fundamental concepts in this context are 
conic combinations and convex combinations.
Conic combinations are linear combinations with positive coefficients. 
Convex combinations are linear combinations
with positive coefficients that sum up to one. Given a set $V$, the set that contains all conic (convex) combinations
of the elements in $V$ is called the conic (convex) hull of $V$. Vice versa, the conic (convex) hull of $V$ is said to be generated
by $V$ under conic (convex) combinations. If $V$ is finite, its convex hull is called finitely generated. These notions give rise
to the two main objects that are studied, cones and polytopes, the first being finitely generated under conic combinations and the second
being finitely generated under convex combinations. The elements of $V$ are called vertices in the case of polytopes. For cones, they are called
rays.

The next useful definition is the Minkowski sum of two sets. Given two sets $A$ and $B$, 
this sum is defined as $A+B = \{a+b | a \in A, b \in B\}$. 
The Minkowski sum of a cone and a polytope is called polyhedron. The sets of rays and vertices that generate the polyhedron are called the
V-representation of the polyhedron. That said, there is a second important represenation of polyhedra: the H-representation. Any
polyhedron is the intersection of a finite number of half-spaces, defined by affine inequalities. The minimal set of half-spaces whose intersection is 
the polyhedron is its H-representation. With every of these half-spaces, we can associate the hyperplane that bounds it. The intersection between
such a hyperplane and the polyhedron is called a facet of the polyhedron. If an inequality defines a half-space such that its bounding hyperplane contains a
facet of a given polyhedron, this inequality is called facet-defining with respect to the polyhedron.
If a polyhedron $P$ has dimension $D$, its facets are the $D-1$ dimensional
polyhedra that together form the boundary of $P$. Intersections of facets are
called faces. Hereby, an $n$-face is a face that has an (affine) dimension of
$n$. Accordingly in the case of a $D$ dimensional polytope, $0$-faces are
vertices, $1$-faces are edges and $(D-1)$-faces are facets.

\subsection{Proof of the main statement}

In the following, we describe a method to find all facet defining inequalities
of a polytope of dimension $D$ that satisfy a set of 
linear equations. However, the possibility of these equations being
inhomogenious makes that task cumbersome. Therefore, we embed the polytope in
a $D+1$ dimensional space by prepending a coordinate to every vertex $\vec v_i$, that is we
define 
\begin{align}
\vec w_i = \binom{1}{\vec v_i}.
\end{align}
We now consider the cone $C$ that is generated by the vectors $\vec w_i$. Note that
we can retain the original polytope by intersecting the cone $C$ with the
hyperplane defined by $x_0 = 1$. Because of its close relationship with the
polytope, we call this cone the corresponding cone  of the polytope. This relationship is not only reflected
in the V-representation. There is also a one-to-one relation between the facet defining inequalities. Let
\begin{align}
\vec x^T \vec b_P \ge -\beta
\end{align}
be a facet defining inequality of the polytope. Then, 
\begin{align}
(1 \; \vec x^T) \binom{\beta}{\vec b_P} \ge 0
\end{align}
is the corresponding facet defining inequality of the cone.

Now assume that we aim to find facets of the polytope that satisfy a given affine constraint
\begin{align}
\vec g^T \vec b_P = \gamma.
\end{align}
Then we can reformulate this as an equivalent linear constraint on the facets of the cone by writing
\begin{align}
(g_0 \; \vec g^T) \binom{\beta}{\vec b_P} = 0
\end{align}
and choosing $g_0 = -\frac{\gamma}{\beta}$. 
If we want to implement an affine constraint that also takes 
the bound of the inequality into account, one can accomplish this by adding
another dimension, simply repeating the previous construction -- this time for
the cone.
From the above discussion we conclude that finding facet defining inequalities of a polyhedron that satisfy affine equality constraints is equivalent to
finding facet defining inequalities of the corresponding cone that satisfy linear equality constraints. We can hence pose the problem as follows: Given a cone
$C$, find all of its facet normal vectors $\vec b_C$ that satisfy 
\begin{align}
G \vec b_C = 0, \label{eq:kernel}
\end{align}
where $G$ is the matrix that captures all linear constraints. 
From \eq{kernel} it  follows that $\vec b_C$ is in the
kernel of $G$. Let
$\ker G = \Span(\{\vec t_j\}_{j=1}^{K})$. We define the matrix $T$
whose $K$ columns of length $D+1$ are made up by the vectors $\vec t_j$, so we have $T_{ij} =
[\vec t_j]_i$. This
allows us to decompose $\vec b_C$ in the basis of $\ker G$ as
    \begin{align}
    \vec b_C = T\vec b_{\tilde C}, \label{eq:btilde}
    \end{align}
where $\vec b_{\tilde C}$ is a vector of dimension $K$.
We will now find a cone $\tilde C$ that has the same dimension as the kernel of $G$,
such that any time $\vec b_C$ defines a facet of
$C$ that satisfies $G \vec b_C = 0$, $\vec b_{\tilde C}$ defines a facet of $\tilde C$. 
Since $\tilde C$ is lower-dimensional compared to $C$ and has at most the same number of rays, it is easier to find the facets of $\tilde C$ than of $C$. The effectivity and feasibility of
the method therefore relies crucially on the dimension $K$ of the kernel of $G$. The 
higher-dimensional and difficult it is to find facets of $C$,
the more linearly independent constraints are needed in order 
to make $\tilde C$ low dimensional and simple enough so we can find its facets.

The following theorem establishes the construction of the cone $\tilde C$.
\begin{theorem}
Let $C = \cone(\{\vec w_i\})$ be a cone and  $\vec b_C$ a facet-normal vector of $C$ that
satisfies $G \vec b_C = 0$ for some matrix $G$. With $T$ and $\vec b_{\tilde C}$ defined as above, we define the cone 
$\tilde C = \cone(\{\vec{\tilde w_i}\})$ of dimension $K$ with $\vec{\tilde
w_i}^T = \vec{w_i}^T T$. Then $\vec b_{\tilde C}$
defines a facet of $\tilde C$.
\end{theorem}
\begin{proof}

We prove the statement in three steps.
(1) The inequality $\vec{\tilde w_i}^T \vec b_{\tilde C} \ge 0$ holds, since \eq{btilde} together with the definition
of the $\vec{\tilde w_i}$ implies 
\begin{align}
\vec{\tilde w_i}^T \vec b_{\tilde C} = \vec{w_i}^T \vec b_C \label{eq:tildenotilde}
\end{align} and $\vec{w_i}^T \vec b \ge 0$ because $\vec b_C$ is facet defining.

(2) The vector $\vec b_{\tilde C}$ defines a face of $\tilde C$, as one can directly see from \eq{tildenotilde}.
With $K = \dim(\ker G)$, the dimension of the face is at most $K-1$, since it is
contained in the $K-1$ dimensional subspace $\{\vec x\,| \,\vec x^T \vec b_{\tilde C} = 0\}$.

(3) The vector $\vec b_{\tilde C}$ defines a facet of $\tilde C$. That is, the dimension of the face is exactly $K-1$. 

Let $B$ be the $M \times(D+1)$ matrix that contains all M rays $\vec w_i$ as rows that
 fulfill $\vec w_i^T \vec b_C = 0$. Since $\vec b_C$ is a facet normal vector,
$B$ has rank $D$. Accordingly, $BT$ is the $M \times K$ matrix that contains all rays $\vec{\tilde w_i}^T$ as 
rows  that fulfill $\vec{\tilde w_i}^T \vec b_{\tilde C} = 0$.
Showing that $\vec{\tilde b}$ defines a facet is equivalent to showing that $\rank(BT) = K-1$. We now prove the latter by contradiction.
Assume there exist two linearly independent vectors
$\vec{\tilde b}, \vec{\tilde c}$ that satisfy $BT\vec{\tilde b} =
BT\vec{\tilde c} = 0$. Thus, $T\vec{\tilde b}$ and $T\vec{\tilde c}$ lie
in the kernel of $B$. Since $\rank(B) = D$, the kernel is one-dimensional, so
we can write $T \vec{\tilde c} = \ell T \vec{\tilde b}$ for some real number
$\ell$. This implies $T(\vec{\tilde c} + \ell \vec{\tilde b}) = 0$. Because
$\vec{\tilde b}$ and $\vec{\tilde c}$ are linearly independent, the kernel
of $T$ has at least dimension one, which is impossible because $T$ has full
column rank.
\end{proof}

The facets of interest of the polytope $P$ can now be found by finding the facets of $\tilde
C$ first, calculating potential facets of $C$ via \eq{btilde}, transforming
these into potential facets of $P$ and finally checking which of the found
inequalities define  facets of $P$.

Note that the method works in the same way for polyhedra. To this end we embed the polyhedron in
a $D+1$ dimensional space by prepending a coordinate to every ray and vertex. 
This coordinate is then set to $0$ for rays and to $1$ for vertices, so the polytope
will lie in the hyperplane defined by $x_0 = 1$. Then the corresponding cone $C$
is the conic hull of all these vectors.

\section{Appendix B: Description of the numerical procedure}

In order to find all generalizations of the I3322 inequality using our
algorithm, we first need to determine its inputs, namely, the rays $\vec
w_i$ of the cone $C$ and the matrix $G$ that captures the linear conditions on
normal vectors of the facets we aim to find. These conditions are established
through the equation $G \vec b_C = 0$. In the following, $b_{ijk}$ denote the
coefficients of the vector $\vec b_C$. We now find the rays of $C$ by first finding
all 512 local deterministic behaviours $\vec v_i$ of the scenario with three
parties, three measurements per party and two outcomes per measurement. Then, we compute the rays of $C$
as
\begin{align}
\vec w_i = \binom{1}{\vec v_i}.
\end{align}
By doing this, we include the trivial correlation $\langle A_0 B_0 C_0\rangle =
1$ as first coordinate.
Now that the rays of $C$ are found, we construct $G$. The first constraint
$G$ is supposed to capture is the symmetry of the Bell inequality under party
permutations. Concretely, the coefficients of the Bell inequality $S =
\sum_{i,j,k} b_{ijk} \langle A_i B_j C_k \rangle$ have to fulfill
\begin{align}
b_{ijk} - b_{jki}& = 0, \label{eq-threeind1}\\
b_{ijk} - b_{kij}& = 0, \\
b_{ijk} - b_{ikj}& = 0, \\
b_{ijk} - b_{kji}& = 0, \\
b_{ijk} - b_{jik}& = 0, \label{eq-threeind2}\\
\intertext{and}
b_{iij} - b_{iji}& = 0, \label{eq-twoind1}\\
b_{iij} - b_{jii}& = 0, \\
b_{jji} - b_{jij}& = 0, \\
b_{jji} - b_{ijj}& = 0, \label{eq-twoind2}
\end{align}
with $i<j<k$. In our scenario, where
the indices take values from $0$ to $3$ we get Eq.~(\ref{eq-threeind1})-(\ref{eq-threeind2}) 
for four different values of $(i,j,k)$ and
Eqs.~(\ref{eq-twoind1})-(\ref{eq-twoind2}) for six different values of $(i,j)$.
This gives us $44$ equations in total. Each of them will be stated as $\vec g_i^T
\vec b_C = 0$ and the vectors $\vec g_i$ then form the rows of the matrix $G$.
Since $\vec b_C$ is a vector in dimension $64$, $G$ has $64$ columns.

Next, we want to ensure that the Bell inequalities we are about to find 
are not only symmetric, but also generalizations of I3322. As explained 
in the main text, there must exist deterministic outcomes $\xi_k = \pm 1$ for 
the measurements on $C$, such that the coefficients of the I3322 inequality 
are obtained via $I^{3322}_{ij} = \sum_k b_{ijk} \xi_k$. So, if a bipartite 
behaviour $\langle A_i B_j\rangle$  saturates I3322, the extended behaviour 
$r_{ijk} = \langle A_i B_j \rangle \xi_k$ saturates its generalization. 
Conversely, if a Bell inequality is saturated by the extended behaviour 
$r_{ijk}$, then $\langle A_i B_j\rangle$ also saturates I3322. Since a
facet is defined by all the vertices it contains, a symmetric and facet 
defining Bell inequality is a generalization of I3322 if and only if for 
all behaviours that saturate I3322 there is local deterministic assignment 
$\vec \xi$, such that the extended behaviours saturate this Bell inequality. 

Since there are three measurements per party and we have the choice between 
two outcomes, we have to take eight possible deterministic assignments into 
account. For each of them, we find the generalized inequalities that can be 
reduced to I3322 by performing the assignment. Hence, we have to run our 
algorithm once for each local deterministic assignment of one party, totalling 
in eight runs in our case. In each run, we have to complete the
matrix $G$ according to the chosen assignment, only the first $44$ rows that
implement the symmetry conditions stay the same. For each run of the algorithm,
we choose one local deterministic assignment $\vec \xi$ and obtain one
extended behaviour $\vec r$ for each behaviour that saturates the I3322
inequality. The extended behaviour then has to saturate a potential
generalization of I3322. This gives us the condition
\begin{align}
\vec r^T \vec b_C = 0. \label{eq-genconstr}
\end{align}
There are 20 local deterministic behaviours that saturate I3322, so we get 20
conditions of type Eq.~(\ref{eq-genconstr}) and the extended
behaviours $\vec r$ make up the next 20 rows of $G$. Therefore, $G$ is a $64
\times 64$ matrix. However, it only has rank 53 because some of the
extended behaviours are already related through the symmetry we implemented
earlier.

The kernel of $G$ can now be found using for example the sympy package in python,
which returns basis vectors with integer coordinates (since $G$ has integer
entries). In our case, the kernel has dimension 11. For each of the 512
rays of $C$, we now obtain a ray of $\tilde C$ in the way described in
Theorem~1. Note, that the map of a ray of $C$ to a ray that lies in $\ker(G)$ is
not necessarily a projection as the basis vectors $\vec t_i$ are not necessarily
normalized. In fact, we want to avoid normalization to preserve the property
that all the objects we deal with only have integer entries. The mapping of rays
of $C$ to rays of $\tilde C$ is not injective. In fact, only 88 rays remain that span the cone $\tilde C$.
The facets of $\tilde C$ can be found within seconds using standard polytope
software such as cdd \cite{cdd}. From these facets, we keep those that
correspond to facets of $C$.

Finally, we compose a list of all valid facets from the eight runs of the
algorithm and remove Bell inequalities that are equivalent to other Bell
inequalities in the list up to relabeling of the measurements or outcomes.
In this way, all 3050 three-party generalizations of I3322 are found.

\section{Appendix C: Constructing LHV models from symmetric quasi-extensions}
\label{app-symext}

In this appendix we first shortly recall the result that a state 
$\varrho_{ABC}$ does not violate any Bell inequality with two 
settings on $B$ and $C$ and an arbitrary number $N$ of measurement 
settings on $A$, if there exists a symmetric extension, that is a 
positive semidefinite operator  $H_{ABB'CC'} = H_{AB'BCC'}$ that 
fulfills $\tr_{B'C'}(H_{ABB'CC'}) = \varrho_{ABC}$. The main idea
stems from \cite{terhal2003}. We focus on the case that is relevant for our paper. 
For convenience, we denote the symmetric extension simply as
$H$, whenever this is possible without causing confusion. 

The argument that links symmetric extensions to local hidden variable models 
goes as follows: Firts, if a symmetric extension exists, we can write, for
any fixed measurement setting $i$ of Alice, a joint probability for both
measurements on Bob and Charlie via
\begin{align}
p(a_i, b_1, b_2, c_1, c_2) = \tr(H  E_{a_i} \otimes E_{b_1} \otimes
E_{b_2} \otimes
E_{c_1} \otimes E_{c_2})
\end{align}
where $E_{a_i}$ is the effect that corresponds to the outcome $a_i$ 
etc. Note that the non-negativity of the operator $H$ is not required, 
it is sufficient if the operator is an entanglement witness, i.e. 
non-negative on product states. Such an operator is called symmetric 
quasi-extension \cite{terhal2003}.

Then, we can define a for all measuremements a joint probability distribution
via \cite{wolf2009}
\begin{align}
p(a_1, & \ldots, a_N,  b_1, b_2, c_1, c_2) \notag\\ = &\frac{p(a_1, b_1, b_2, c_1, c_2)
\cdots p(a_N, b_1, b_2, c_1, c_2)}{p(b_1, b_2, c_1, c_2)^{N-1}}.
\end{align}
Finally, it is well established that if such a joint probability 
distribution exists, then a LHV model exists and no Bell inequality 
can be violated \cite{fine1982, abramsky2011}.

In our work, we make use of this connection to find a three-qubit state 
$\varrho_{ABC}$ that violates one of the Bell inequalities $F_i$ in 
Eqs.~(\ref{eq-bi1})-(\ref{eq-bi3}), but which is not violating any 
three-partite Bell inequality with two settings per party and, in addition, 
it is not violating any bipartite Bell inequality (such as I3322) with any of 
its two-body marginals. We achieve this by demanding that the state $\varrho_{ABC}$ 
possesses a symmetric extension and separable two-body marginals. Since
the marginals are two-qubit states, the separability condition can be 
implemented using the criterion of the positivity of the partial transpose 
\cite{horodecki1996,peres1996}. 

In this way, all of the constraints are semidefinite constraints. Maximizing 
the violation of a Bell inequality under these constraints is a semidefinite 
program, if the measurement settings are given. As initially no measurement 
settings are given, we pick random measurement settings for the three parties 
before optimizing over the state. Then, keeping the state fixed, each optimization 
over the measurement settings of a single party is again a semidefinite programm. 
We then alternate between these two steps -- the optimization of the state on the 
one hand and the optimization of the measurement settings on the other -- in a 
seesaw algorithm. We solve the semidefinite programs with Mosek \cite{mosek} 
through Picos \cite{picos}. 

Typically, after 50 iterations a convergence is reached.  In this way, we 
find a state $\varrho_{ABC}^{(1)}$ with $F_1 \approx -0.065$, a state 
$\varrho_{ABC}^{(2)}$ with $F_2 \approx -0.043$ and a state $\varrho_{ABC}^{(3)}$
with $F_3 \approx -0.063$. The symmetric extensions and measurement
settings that lead to these violations can be found online.
\footnote{The files {\tt symm-extensions.txt}, {\tt F1\_H.txt}, {\tt F2\_H.txt}, and {\tt F3\_H.txt}
are included in the source files of this arxiv submission.}.

\bibliography{biblio}

\end{document}